\documentclass[preprint]{elsarticle}
\usepackage{amsthm}
\usepackage{amssymb}
\usepackage{amsmath}
\usepackage{graphicx}
\usepackage{lineno}

\newtheorem{theorem}{Theorem}[section]
\newtheorem{lemma}[theorem]{Lemma}
\newtheorem{proposition}[theorem]{Proposition}
\newtheorem{corollary}[theorem]{Corollary}
\newtheorem{definition}[theorem]{Definition}
\newtheorem{defn}[theorem]{Definition}
\newtheorem{example}[theorem]{Example}
\newtheorem{remark}[theorem]{Remark}

\def\operatorname#1{{\mathop{\rm #1}}}
\newcommand{\ord}{\operatorname{ord}}

\def\tao{\tau_o}

\def\too{\stackrel{(o)}{\longrightarrow}}
\def\totau{\stackrel{\tao}{\longrightarrow}}

\def\comp{\leftrightarrow}
\def\tophi{\stackrel{\tau_{\Phi}}{\longrightarrow}}

\newcounter{ok}
{\end{list}}

\newcounter{aok}
{\end{list}}

\def\go#1;#2;#3 {\vbox to0pt{\kern-#3\hbox{\kern#2 #1}\vss}\nointerlineskip}

\begin{document}
\begin{frontmatter}
\title{Almost orthogonality and Hausdorff interval topologies
of atomic lattice effect algebras}

\author[Paseka]{Jan Paseka\corref{cor1}\fnref{fn1}}
\cortext[cor1]{Corresponding author}
\ead{paseka@math.muni.cz}
\fntext[fn1]{This work was supported by
by  the  Grant Agency of the Czech
Re\-pub\-lic un\-der the grant
No.~201/06/0664 and by
the  Ministry of Education of the Czech Republic
under the project MSM0021622409}
\author[Riecan]{Zdenka Rie\v{c}anov\'a\fnref{fn2}}
\ead{zdenka.riecanova@stuba.sk}
\fntext[fn1]{This work was supported by the Slovak Resaerch and Development
Agency under the contract No. APVV--0071--06 and the grant
VEGA-1/3025/06 of M\v{S}~SR.}

\address[Paseka]{Department of Mathematics
 and Statistics,
Faculty of Science,
Masaryk University,
{Kotl\'a\v r{}sk\' a\ 2},
CZ-611~37~Brno, Czech Republic}
\address[Riecan]{Department of Mathematics,
Faculty of Electrical Engineering and Information Technology,
Slovak University of Technology, Ilkovi\v{c}ova~3,
SK-812~19~Bratislava, Slovak Republic}
\address[Wu Junde]{Department of Mathematics, Zhejiang University,
Hangzhou 310027, People's Republic of China}
\author[Wu Junde]{Wu Junde}
\ead{wjd@zju.edu.cn}

\date{\today}





\begin{abstract}
We prove that the interval topology of an
Archimedean atomic lattice effect algebra $E$ is Hausdorff
whenever the set of all atoms of $E$ is almost orthogonal.
In such a case $E$
is order continuous. If moreover $E$ is complete
then order convergence of nets of elements of $E$ is topological
and hence it coincides with convergence in the
order topology and this topology is compact Hausdorff compatible
with a uniformity induced by a separating function family
on $E$ corresponding to compact and cocompact elements.
For block-finite Archimedean atomic lattice effect algebras the
equivalence of almost orthogonality and s-compact generation
is shown. As the main application we obtain the state smearing
theorem for these effect algebras, as well as the continuity
of $\oplus$-operation in the order and interval topologies
on them.
\end{abstract}


%

\end{frontmatter}

\section{Introduction, basic definitions and facts}

\label{intro}

In the study of effect algebras (or more general, quantum structures)
as carriers of states and probability measures, an important tool
is the study of topologies on them. We can say that topology is practically
equivalent with the concept of convergence. From the probability
point of view the convergence of nets is the main tool in spite of that
convergence of filters is easier to handle and preferred in
the modern topology. It is because states or probabilities are
mappings (functions) defined on elements but not on subsets of
quantum structures. Note also, that connections between order
convergence of filters and nets are not trivial.
For instance, if a filter order converges to
some point of a poset then the associated net need not order
converge (see e.g, \cite{ZR19}).

On the other hand certain topological properties
of studied structures characterize also their certain
algebraic properties and conversely. For instance a known fact
is that a Boolean algebra $B$ is atomic iff
the interval topology $\tau_i$ on $B$ is Hausdorff
(see \cite[Corollary 3.4]{ZR47}).
This is not more valid for lattice effect algebras (even
MV-algebras). By Frink's Theorem the interval topology
$\tau_i$ on $B$ (more generally on any lattice $L$) is
compact iff it is a complete lattice \cite{frink}.
In \cite{PR4} it was proved that if a lattice
effect algebra $E$ (more generally any basic algebra) is compactly
generated then $E$ is atomic.

We are going to prove
that on an Archimedean atomic lattice effect algebra $E$ the interval
topology $\tau_i$ is Hausdorff and $E$ is (o)-continuous if and only if
$E$  is almost orthogonal.
Moreover, if $E$ is complete then $\tau_i$ is compact
and coincides with the order topology $\tau_o$ on $E$ and
this compact topology $\tau_i=\tau_o$ is
compatible with a uniformity on $E$ induced by a separating
function family on $E$ corresponding to compact and cocompact elements
of $E$.

As the main corollary of that we obtain that every Archimedean
 atomic block-finite
lattice effect algebra $E$ has  Hausdorff interval topology and hence
both topologies $\tau_i$ and $\tau_o$ are Hausdorff
and they coincide. In this case
almost orthogonality of $E$ and s-compact generation
by finite elements of $E$ are equivalent.
As an application the state smearing
theorem for these effect algebras is formulated.
Moreover, the continuity
of $\oplus$-operation in $\tau_i$ and $\tau_o$
on them is shown.

\begin{defn}\label{def:EA}
A partial algebra $(E;\oplus,0,1)$ is called an effect algebra if
$0$, $1$ are two distinct elements and $\oplus$ is a partially
defined binary operation on $E$ which satisfy the following
conditions for any $a,b,c\in E$:
\begin{description}
\item[\rm(Ei)\phantom{ii}] $b\oplus a=a\oplus b$ if $a\oplus b$ is defined,
\item[\rm(Eii)\phantom{i}] $(a\oplus b)\oplus c=a\oplus(b\oplus c)$  if one
side is defined,
\item[\rm(Eiii)] for every $a\in E$ there exists a unique $b\in
E$ such that $a\oplus b=1$ (we put $a'=b$),
\item[\rm(Eiv)\phantom{i}] if $1\oplus a$ is defined then $a=0$.
\end{description}
\end{defn}

We often denote the effect algebra $(E;\oplus,0,1)$ briefly by
$E$. In every effect algebra $E$ we can define the partial order
$\le$ by putting

\noindent{}
$a\le b$ \mbox{ and } $b\ominus a=c$ \mbox{ iff }$a\oplus c$
\mbox{ is defined and }$a\oplus c=b$\mbox{, we set }$c=b\ominus a$\,.

If $E$ with the defined partial order is a lattice (a complete
lattice) then $(E;\oplus,0,1)$ is called a {\em lattice effect
algebra} ({\em a complete lattice effect algebra}).

Recall that a set $Q\subseteq E$ is called a {\em sub-effect algebra}
of the effect algebra $E$ if
\begin{enumerate}
\item[(i)] $1\in Q$
\item[(ii)] if out of elements $a,b,c\in E$ with $a\oplus b=c$
two are in $Q$, then $a,b,c\in Q$.
\end{enumerate}

If $Q$ is simultaneously a sublattice of $E$ then $Q$ is called
a {\em sub-lattice effect algebra} of $E$.

We say that a finite system $F=(a_k)_{k=1}^n$ of not necessarily
different elements of an effect algebra $(E;\oplus,0,1)$ is
{\em $\oplus$-orthogonal} if $a_1\oplus a_2\oplus\dots\oplus a_n$
(written $\bigoplus\limits_{k=1}^na_k$ or $\bigoplus F$) exists
in $E$. Here we define $a_1\oplus a_2\oplus\dots\oplus a_n=
(a_1\oplus a_2\oplus\dots\oplus a_{n-1})\oplus a_n$ supposing
that $\bigoplus\limits_{k=1}^{n-1}a_k$ exists and
$\bigoplus\limits_{k=1}^{n-1}a_k\le a'_n$. An arbitrary system
$G=(a_{\kappa})_{\kappa\in H}$ of not necessarily
different elements of $E$ is
{\em $\oplus$-orthogonal} if  $\bigoplus K$ exists for every finite
$K\subseteq G$. We say that for a $\oplus$-orthogonal system
$G=(a_{\kappa})_{\kappa\in H}$ the element $\bigoplus G$ exists
iff $\bigvee\{\bigoplus K\mid K\subseteq G$, $K$ is finite$\}$
exists in $E$  and then we put
$\bigoplus G=\bigvee\{\bigoplus K\mid K\subseteq G\}$
(we write $G_1\subseteq G$ iff there is $H_1\subseteq H$ such
that $G_1=(a_{\kappa})_{\kappa\in H_1}$).

Recall that elements $x$ and $y$ of a lattice effect algebra are
called {\em compatible} (written $x\leftrightarrow y$) if $x\vee
y=x\oplus(y\ominus(x\wedge y))$ \cite{Kop2}. For $x\in E$ and
$Y\subseteq E$ we write $x\leftrightarrow Y$ iff $x\leftrightarrow
y$ for all $y\in Y$. If every two elements are compatible then $E$
is called an {\em MV-effect algebra}. In fact, every MV-effect algebra
can be organized  into an MV-algebra (see \cite{C.C.Ch}) if we extend the partial
$\oplus$ to a total operation by setting $x\oplus
y=x\oplus(x'\wedge y)$ for all $x,y\in E$ (also conversely,
restricting a total $\oplus$ into partial $\oplus$ for only $x,y\in
E$ with $x\le y'$ we obtain a MV-effect algebra).

Moreover, in \cite{ZR56} it was proved that every lattice effect
algebra is a set-theoretical union of MV-effect algebras called
blocks. {\em Blocks} are maximal subsets of pairwise compatible
elements of $E$, under which every subset of pairwise compatible
elements is by Zorn's Lemma contained in a maximal one. Further,
blocks are sub-lattices and sub-effect algebras of $E$ and hence
maximal sub-MV-effect algebras of $E$. A lattice effect algebra
is called {\em block-finite} if it has only finitely many blocks.

Finally note that lattice effect algebras generalize orthomodular
lattices \cite{K} (including Boolean algebras) if we assume
existence of unsharp elements $x\in E$, meaning that $x\wedge
x'\ne0$. On the other hand the {\em set} $S(E)=\{x\in E\mid
x\wedge x'=0\}$ {\em of all sharp elements} of a lattice effect
algebra $E$ is an orthomodular lattice \cite{ZR57}. In this sense
a lattice effect algebra is a ``smeared'' orthomodular lattice,
while an MV-effect algebra is a  ``smeared'' Boolean algebra. An
orthomodular lattice $L$ can be organized into a lattice effect
algebra by setting $a\oplus b=a\vee b$ for every pair $a,b\in L$
such that $a\le b^{\perp}$.

For an element $x$ of an effect algebra $E$ we write
$\ord(x)=\infty$ if $nx=x\oplus x\oplus\dots\oplus x$ ($n$-times)
exists for every positive integer $n$ and we write $\ord(x)=n_x$
if $n_x$ is the greatest positive integer such that $n_xx$
exists in $E$.  An effect algebra $E$ is {\em Archimedean} if
$\ord(x)<\infty$ for all $x\in E$. We can show that every
complete effect algebra is Archimedean (see \cite{ZR54}).

An element $a$ of an effect algebra $E$ is an {\em atom} if $0\le
b<a$ implies $b=0$ and $E$ is called {\em atomic} if for every
nonzero element $x\in E$ there is an atom $a$ of $E$ with $a\le
x$. If $u\in E$ and either $u=0$ or $u=p_1\oplus
p_2\oplus \dots\oplus p_n$ for  some not
necessarily different atoms $p_1,p_2, \dots,p_n\in E$ then
$u\in E$ is called {\em finite\/} and $u'\in E$ is called {\em cofinite\/}.
If $E$ is a lattice effect algebra then for $x\in E$ and an
atom $a$ of $E$ we have $a\leftrightarrow x$ iff $a\le x$ or $a\le
x'$. It follows that  if $a$ is an atom of a block $M$ of $E$ then
$a$ is also an atom of $E$. On the other hand if $E$ is atomic
then, in general, every block in $E$ need not be atomic (even for
orthomodular lattices \cite{Bel&Cas}).

The following theorem is well known.

\begin{theorem}{\rm\cite[Theorem 3.3]{ZR65}}\label{Theorem3.3}
Let $(E;\oplus,0,1)$ be an Archimedean atomic lattice effect
algebra. Then to every nonzero element $x\in E$ there are
mutually distinct atoms $a_{\alpha}\in{E}$, $\alpha{\in\mathcal E}$ and
positive integers $k_{\alpha}$ such that
$$
x=\bigoplus\{k_\alpha a_\alpha\mid\alpha\in{\mathcal E}\}
=\bigvee\{k_\alpha a_\alpha\mid\alpha\in{\mathcal E}\}
$$
under which $x\in S(E)$ iff
$k_\alpha=n_{a_\alpha}=\ord(a_\alpha)$ for all $\alpha\in {\mathcal E}$.
\end{theorem}

\begin{definition} \mbox{\rm{}(1)} An element $a$ of a lattice $L$ is called
{\em compact} iff, for any $D\subseteq L$ with $\bigvee D\in L$, if $a\leq \bigvee D$
then $a\leq \bigvee F$ for some finite $F\subseteq D$.

\mbox{\rm{}(2)} A lattice $L$ is called {\em compactly generated} iff every
element of $L$ is a join of compact elements.
\end{definition}

The notions of {\em cocompact element} and {\em cocompactly generated lattice}
can be defined dually. Note that compact elements are important in
computer science in the semantic approach called domain theory, where
they are considered as a kind of primitive elements.


\section{Characterizations of interval topologies on bounded lattices}
\label{Interval topology}

The order convergence of nets ((o)-convergence), interval topology
$\tau_i$ and order-topology $\tau_o$ ((o)-topology) can be defined
on any poset. In our observations we will consider only bounded lattices
and we will give a characterization of interval topologies on them.

\begin{definition}
Let $L$ be a bounded lattice. Let
$\mathcal{H}=\{[a,b]\subseteq L|a,b\in L$ with $a\le
b\}$ and let $\mathcal{G}=\{\bigcup^{n}_{k=1}[ a_k,b_k]|[
a_k,b_k]\in\mathcal{H}, k=1,2,...,n\}$. The \emph{interval topology
$\tau_i$ of $L$} is the topology of $L$ with $\mathcal{G}$ as a
closed basis, hence with $\mathcal{H}$ as a closed subbasis.
\end{definition}

From definition of $\tau_i$ we obtain that $U\in\tau_i$ iff for each
$x\in U$ there is $F\in\mathcal{G}$ such that $x\in L\backslash
F\subseteq U$.

\begin{definition} Let $L$ be a poset.
\begin{itemize}
\item[\mbox{\rm{}(i)}] A net $(x_\alpha)_{\alpha\in\mathcal{E}}$
of elements of $L$ {\emph{ order converges}}
({\emph{$(o)$-converges}}, for short) to a point $x\in L$
if there exist nets
$(u_\alpha)_{\alpha\in\mathcal{E}}$ and
$(v_\alpha)_{\alpha\in\mathcal{E}}$ of elements of $L$
such that
$$
x\uparrow u_\alpha\le x_\alpha\le v_\alpha\downarrow x,\, {\alpha\in\mathcal{E}}
$$
\noindent{}where $x\uparrow u_\alpha$ means that $u_{\alpha_1}\le u_{\alpha_2}$ for every
${\alpha_1}\le {\alpha_2}$ and
$x=\bigvee\{x_\alpha\mid\alpha\in\mathcal{E}\}$. The meaning of
$v_\alpha\downarrow x$ is dual.

We write  $x_\alpha\stackrel{(o)}{\rightarrow} x, {\alpha\in\mathcal{E}}$ in $L$.

\item[\mbox{\rm{}(ii)}]  A topology $\tau_{0}$ on $L$ is called
the \emph{order topology}  on $L$ iff
\begin{enumerate}
\item[\mbox{\rm{}(a)}]  for any net
$(x_\alpha)_{\alpha\in\mathcal{E}}$ of elements of $L$ and $x\in L$:
$x_\alpha\stackrel{(o)}{\rightarrow} a$ in L $\Rightarrow$
{$x_\alpha\stackrel{\tau_{0}}{\rightarrow} x$}, ${\alpha\in\mathcal{E}}$, where
{$x_\alpha\stackrel{\tau_{0}}{\rightarrow} x$} denotes that
$(x_\alpha)_{\alpha\in\mathcal{E}}$ \emph{converges} to $x$ in
the topological space $(L,\tau_{0})$,
\item[\mbox{\rm{}(b)}] if $\tau$ is a topology on $L$ with
property (a) then $\tau\subseteq \tau_o$.
\end{enumerate}
\end{itemize}
\end{definition}
Hence $\tau_o$ is the strongest (finest, biggest) topology on $L$
with property (a).

Recall that, for  a directed set $(\mathcal{E},\leq)$, a subset
$\mathcal{E}'\subseteq\mathcal{E}$ is called
\emph{cofinal} in $\mathcal{E}$ iff for every $\alpha\in\mathcal{E}$
there is $\beta\in\mathcal{E}'$ such that $\alpha\leq\beta$.
A special kind of a {\em subnet} of  a net
$(x_\alpha)_{\alpha\in\mathcal{E}}$ is net
$(x_\beta)_{\beta\in\mathcal{E}'}$ where $\mathcal{E}'$ is
a cofinal subset of $\mathcal{E}$. This kind of subnets works
in many cases of our considerations.

In what follows we often use the following useful
characterization of topological convergence of nets:

\begin{lemma}\label{charcofin} For a net $(x_\alpha)_{\alpha\in\mathcal{E}}$ of elements of
a topological space $(X, \tau)$ and $x\in X$:
\begin{center}
\begin{tabular}{l c l}
$x_\alpha\stackrel{\tau}{\rightarrow}x,
\alpha\in\mathcal{E}$ &iff& for all
$\mathcal{E}'\subseteq\mathcal{E}$, where $\mathcal{E}'$ is
cofinal in $\mathcal{E}$ there exist\\
&&$\mathcal{E}''\subseteq\mathcal{E}'$, $\mathcal{E}''$ cofinal in
$\mathcal{E}'$ such that $x_\gamma\stackrel{\tau}{\rightarrow}x,
\gamma\in\mathcal{E}''$.
\end{tabular}
\end{center}
\end{lemma}
\begin{proof} $\Rightarrow$: It is trivial.

\noindent{}$\Leftarrow$: Let for every $\mathcal{E}'\subseteq\mathcal{E}$, where
$\mathcal{E}'$ is cofinal in $\mathcal{E}$ there exist
$\mathcal{E}''\subseteq\mathcal{E}'$, $\mathcal{E}''$ cofinal in
$\mathcal{E}'$ and $x_\gamma\stackrel{\tau}{\rightarrow}x,
\gamma\in\mathcal{E}''$, and let
$x_\alpha\not\stackrel{\tau}{\rightarrow}x, \alpha\in\mathcal{E}$.
Then there exist $U(x)\in\tau$ such that for all
$\alpha\in\mathcal{E}$ there exist $\beta_\alpha\in\mathcal{E}$ with
$\beta_\alpha\ge\alpha$ and $x_{\beta_\alpha}\not\in U(x)$. Let
$\mathcal{E}'=\{\beta_\alpha\in\mathcal{E}|\alpha\in\mathcal{E},\beta_\alpha\ge\alpha,
x_{\beta_\alpha}\not\in U(x)\}$ then
$x_{\beta_\alpha}\not\stackrel{\tau}{\rightarrow}x,
{\beta_\alpha}\in\mathcal{E}'$ and for all cofinal
$\mathcal{E}''\subseteq\mathcal{E}'$:
$x_\gamma\not\stackrel{\tau}{\rightarrow}x, \gamma\in\mathcal{E}''$.
Hence there exists $\mathcal{E}'\subseteq\mathcal{E}$ cofinal in
$\mathcal{E}$ and for all $\mathcal{E}''\subseteq\mathcal{E}'$,
$\mathcal{E}''$ cofinal in $\mathcal{E}'$:
$x_\gamma\not\stackrel{\tau}{\rightarrow}x, \gamma\in\mathcal{E}''$
a contradiction.
\end{proof}

Further, let us recall the following well known facts:

\begin{lemma}\label{pomlem} Let $L$ be a bounded lattice. Then
\begin{enumerate}
\item[\mbox{\rm{}(i)}]  $F\subseteq L$ is $\tau_0$-closed iff for every net
$(x_\alpha)_{\alpha\in\mathcal{E}}$ of elements of $L$ and $x\in L$:\\
$(x_\alpha\in
F, x_\alpha\stackrel{(o)}{\rightarrow} x,
\alpha\in\mathcal{E}) \Rightarrow x\in F$.
\item[\mbox{\rm{}(ii)}]   For every
$a,b\in L$ with $a\le b$ the interval $[a,b]$ is $\tau_0$-closed.
\item[\mbox{\rm{}(iii)}] $\tau_{i} \subseteq \tau_{o}$.
\item[\mbox{\rm{}(iv)}] For any net $(x_\alpha)_{\alpha\in\mathcal{E}}$ of elements of
a $L$ and $x\in L$:
$$x_\alpha\stackrel{(o)}{\rightarrow} x, {\alpha\in\mathcal{E}}\
\Longrightarrow\
x_\alpha\stackrel{\tau_{i}}{\rightarrow} x, {\alpha\in\mathcal{E}}.$$
\item[\mbox{\rm{}(v)}] If $\tau_i$  is Hausdorff then $\tau_0=\tau_i$ (see \cite{erne1}).
\item[\mbox{\rm{}(vi)}] The interval topology $\tau_i$ of a lattice $L$ is
compact iff $L$ is a complete lattice (see \cite{frink}).
\end{enumerate}
\end{lemma}

Finally, let us note that compact Hausdorff topological space is
always normal. Thus separation axiom $T_2,T_3$ and $T_4$ are
trivially equivalent for the interval topology of a complete
lattice $L$.

\begin{theorem}\label{compconvti}
Let $L$ be a complete lattice with interval
topology $\tau_i$. If $F\subseteq L$ is a complete sub-lattice of $L$
then
\begin{enumerate}
\item[\mbox{\rm{}(a)}]   $\tau^{F}_{i}=\tau_i\cap F$ is the interval topology of $F$,
\item[\mbox{\rm{}(b)}]  for any net $(x_\alpha)_{\alpha\in\mathcal{E}}$ of elements of
$F$ and $x\in F$:
$$x_\alpha\stackrel{\tau^{F}_{i}}{\rightarrow} x, {\alpha\in\mathcal{E}}
\ \Longleftrightarrow \
x_\alpha\stackrel{\tau_{i}}{\rightarrow} x, {\alpha\in\mathcal{E}}.$$
\end{enumerate}
\end{theorem}
\begin{proof}\mbox{\rm{}(a)}: Let $\mathcal{H}$ and $\mathcal{H}_{F}$ be a
closed subbasis of $\tau_i$ and $\tau^{F}_{i}$ respectively. Then
evidently
$\mathcal{H}\cap F=\{[a,b]\cap F|[a,b]\in\mathcal{H}\}$
is a closed subbasis of $\tau_i\cap
F$. Further for $[c,d]_F\in\mathcal{H}_F$ we have
$[c,d]_F=\{x\in F|c\le x\le d\}=[c,d]\cap F\in\mathcal{H}\cap F$.
Conversely, since $F$ is a
complete sub-lattice of $L$, if $[a,b]\in\mathcal{H}$
then $[a,b]\cap F=\{x\in F|a\le x\le b\}$  and either
$[a,b]\cap F=\emptyset$ or there exist
$c=\wedge\{x\in F|a\le x\le b\}$ and $d=\vee\{x\in F|a\le x\le
d\}$ and $[a,b]\cap F=[c,d]_F\in\mathcal{H}_F$. This proves that
$\tau^{F}_{i}=\tau_i\cap F$.

\noindent{}\mbox{\rm{}(b)}: This is an easy consequence of \mbox{\rm{}(a)}.
\end{proof}

\section{Hausdorff interval topology of almost orthogonal
Archimedean atomic lattice effect algebras
and their order continuity}

The atomicity of Boolean algebra $B$ is equivalent with
Hausdorffness of interval topology on $B$
(see \cite{katetov}, \cite{sarymsakov} and \cite[Corollary 3.4]{ZR47}).
This is not more valid for lattice effect algebras, even
also for MV-algebras.

\begin{example} {\rm Let $M=[0, 1]\subseteq {\mathbb R}$
be a standard MV-effect algebra, i.e., we define
$a\oplus b=a+b$ iff $a+b\leq 1$, $a, b\in M$.
Then $M$ is a complete (o)-continuous lattice with
$\tau_i=\tau_o$ being Hausdorff and with
(o)-convergence of nets coinciding with
$\tau_o$-convergence. Nevertheless, $M$ is not atomic.
}
\end{example}

We have proved in \cite{PR4} that a complete lattice
effect algebra is atomic and (o)-continuous lattice
iff $E$ is compactly generated. Nevertheless, in such
a case, the interval topology on $E$ need not be Hausdorff.

\begin{example}{\rm Let $E$ be a horizontal sum of infinitely many finite
 chains $(P_i,\bigoplus_i, 0_i, 1_i)$ with at least 3 elements,
 $i=1,2,...,n, \dots, $ (i. e., for
 $i=1,2,...,n, \dots,$, we identify all $0_i$ and all
 $1_i$ as well, $\bigoplus_i$ on $ P_i$
 are preserved and any
 $a\in P_i\backslash\{0_i,1_i\}$, $b\in
 P_j\backslash\{0_j,1_j\}$ for $i\not=j$
 are noncomparable). Then  $E$ is an atomic complete lattice effect algebra,
$E$ is not block-finite
 and the interval topology $\tau_i$ on $E$  is compact. 
 Nevertheless, $\tau_i$ is not Hausdorff
because e.g., for $a\in P_i, b\in P_j, i\not=j$,
$a, b$ noncomparable, we have
 $[a,1]\cap [0, b]=\emptyset$ and there is no finite family
 $\mathcal{I}$ of closed intervals in $E$ separating
 $[a,1], [0, b]$ (i.e., the lattice
 $E$ can not be covered by a finite number of closed intervals from $\mathcal{I}$ each of
 which is disjoint with at least one of the intervals $[a,1]$
 and $[0,b]$). This implies that $\tau_i$ is not Hausdorff
 by \cite[Lemma 2.2]{ZR47}.
 Further $E$ is {\em compactly generated by finite elements}
 (hence (o)-continuous). It follows by \cite{PR4} that
 the order topology $\tau_o$ on $E$ is a uniform topology
 and (o)-convergence of nets on $E$ coincides with
 $\tau_o$-convergence.
 }
 \end{example}

 In what follows we shall need an extension of
\cite[Lemma 2.1 (iii)]{ZR70}.

\begin{lemma}\label{extfu} Let $E$ be a lattice effect algebra, $x, y\in E$.
Then $x\wedge y=0$ and $x\leq y'$ iff $kx\wedge ly=0$ and $kx\leq (ly)'$,
whenever $kx$ and $ly$ exist in $E$.
\end{lemma}
\begin{proof} Let $x\leq y'$, $x\wedge y=0$ and $2y$ exists in $E$. Then
$x\oplus y=(x\vee y)\oplus (x\wedge y)=x\vee y\leq y'$ and hence there exists
$x\oplus 2y=(x\vee y)\oplus y=(x\oplus y)\vee 2y=x\vee y\vee 2y=x\vee  2y$,
which gives that $x\leq (2y)'$ and $x\wedge 2y=0$. By induction, if
$ly$ exists then $x\oplus ly=x\vee  ly$ and hence $x\leq (ly)'$ and $x\wedge ly=0$.

Now, $x\leq (ly)'$ iff $ly\leq x'$ and because $x\wedge ly=0$, we obtain
by the same argument as above that $ly\oplus kx=ly\vee  kx$, hence
$kx\leq (ly)'$ and $ly\wedge kx=0$ whenever $kx$  exists in $E$.

Conversely, $kx\wedge ly=0$ implies that $x\wedge y=0$ and $kx\leq (ly)'$
implies  $x\leq kx\leq (ly)'\leq y'$. \end{proof}

In next we will use the statement of Lemma \ref{extfu} in the following form:
For any $x, y\in E$ with $x\wedge y=0$, $x\not\leq y'$ iff $kx\not\leq (ly)'$,
whenever $kx$ and $ly$ exist in $E$.

\begin{definition}\label{almoml} Let $E$ be an atomic lattice effect algebra.
 $E$ is said to be {\em almost orthogonal} if the set
$\{ b\in E \mid b\not\leq a', b\ \mbox{is an atom}\}$ is finite for every
atom $a\in E$.
 \end{definition}

Note that our definition of almost orthogonality coincides with the usual
 definition for orthomodular lattices (see e.g. \cite{ZR26, ZR34}).

\begin{theorem}\label{almdef} Let $E$ be an Archimedean atomic
lattice effect algebra.
 Then $E$ is  {almost orthogonal} if and only if
for any atom $a\in E$ and any integer $l$,
$1\leq l\leq n_a$,  there are finitely many atoms $c_1, \dots, c_m$ and
integers $j_1, \dots, j_m$,
$1\leq j_1\leq n_{c_1}, \dots, 1\leq j_m\leq n_{c_m}$ such that
$j_k{}c_k\not\leq (la)'$ for all $k\in \{ 1, \dots, m\}$ and, for all
$x\in E$,
$x\not\leq (la)'$ implies $j_{k_0}{}c_{k_0}\leq x$ for some
${k_0}\in \{ 1, \dots, m\}$.
 \end{theorem}
 \begin{proof}$\Longrightarrow$: Assume that $E$ is  {almost orthogonal}.
 Let $a\in E$ be an atom, $1\leq l\leq n_a$. We shall denote
 $A_a=\{b\in E \mid b \ \mbox{is}\ \mbox{an}\ \mbox{atom},
b\not\leq a'\}$. Clearly, $A_a$ is finite i.e.
$A_a=\{ b_1, \dots, b_n\}$ for suitable atoms $b_1, \dots, b_n$
from $E$.

Let $b\in E$ be an atom, $1\leq k\leq n_b$ and $kb\not \leq (la)'$.
Either $b=a$ or $b\not={}a$ and in this case we have by
Lemma \ref{extfu} (iv) that $b\not \leq a'$. Hence
either $b=a$ or $b\in A_a$. Let us put
$\{ c_1, \dots, c_m\}= \left\{
     \begin{array}{l l} A_a&\mbox{if}\ a\in S(E)\\ 
                        A_a\cup \{ a\}&\mbox{otherwise}
     \end{array}\right.$. In both cases we have that
     $a\in \{ c_1, \dots, c_m\}$.

Now, let $x\in E$ and $x\not \leq (la)'$.
     By Theorem \ref{Theorem3.3} there is an atom $c\in E$ and
     an integer $1\leq j\leq n_c$ such that $jc\leq x$ and $jc\not \leq (la)'$.
     Either $c=a$ or $c\not\leq a$. In the first case we have that
     $j\geq (n_a-l+1)$ i.e. $x\geq  (n_a-l+1)a$. In the second case we
     get that $c\not \leq a'$ i.e. $c\in A_a$ and $x\geq b_i$ for
     suitable $i\in \{ 1, \dots, n\}$. Hence it is enough to
     put $j_k=1$ if $c_k\in A_a$ and $j_k=(n_a-l+1)$ if $c_k=a$.

\noindent{}$\Longleftarrow$: Conversely, let $a\in E$ be an atom.
Then there are finitely many atoms $c_1, \dots, c_m$ and
integers $j_1, \dots, j_m$,
$1\leq j_1\leq n_{c_1}, \dots, 1\leq j_m\leq n_{c_m}$ such that
$j_k{}c_k\not\leq a'$ for all $k\in \{ 1, \dots, m\}$ and, for all
$x\in E$,
$x\not\leq a'$ implies $j_{k_0}{}c_{k_0}\leq x$ for some
${k_0}\in \{ 1, \dots, m\}$. Let us check that
$A_a\subseteq \{ c_1, \dots, c_m\}$. Let $b\in A_a$. Then
$b\geq j_{k_0}{}c_{k_0}\geq c_{k_0}$ for some
${k_0}\in \{ 1, \dots, m\}$. Hence $b=c_{k_0}$. This yields
$A_a$ is finite.
\end{proof}


\begin{lemma}\label{almostort} Let $E$ be an almost
orthogonal Archimedean atomic lattice effect algebra.
Then, for any atom $a\in E$ and any integer $l$,
$1\leq l\leq n_a$ there are finitely many atoms
$b_1, \dots, b_n$ and integers $j_1, \dots, j_n$,
$1\leq j_1\leq n_{b_1}, \dots, 1\leq j_n\leq n_{b_n}$ such that
\begin{center}
\begin{tabular}{@{}l l}
&$E=[0, (la)']\cup \left(\bigcup_{k=1}^{n}[j_k{}b_k, 1]\cup %
[(n_a+1-l)a, 1]\right)$\\
and\qquad\qquad&\\
&$[0, (la)']\cap \left(\bigcup_{k=1}^{n}[j_k{}b_k, 1]\cup %
[(n_a+1-l)a, 1]\right)=\emptyset$.
\end{tabular}
\end{center} Hence $[0, (la)']$ is a clopen subset
in the interval topology.
\end{lemma}
\begin{proof} Let $a\in E$ be an atom, $1\leq l\leq n_a$.
By Definition \ref{almdef},  let $\{j_1{}b_1, \dots, j_n{}b_n\}$ be
the finite set of  non-orthogonal
finite elements to $la$ of the form $j_kb_k$,
$1\leq j_k\leq n_{b_k}$ minimal  such that $b_1, \dots, b_n$
are atoms different from $a$. We put
$D=[0, (la)']\cup \left(\bigcup_{k=1}^{n}[j_kb_k, 1]\cup %
[(n_a+1-l)a, 1]\right)$. Let us check that $D=E$.
Clearly, $D\subseteq E$. Now, let $z\in E$. Then by
Theorem \ref{Theorem3.3} there are
mutually distinct atoms $c_{\gamma}\in{E}$, $\gamma{\in\mathcal E}$ and
integers $t_{\gamma}$ such that
$$
z=\bigoplus\{t_\gamma c_\gamma\mid\gamma\in{\mathcal E}\}
=\bigvee\{t_\gamma c_\gamma\mid\gamma\in{\mathcal E}\}.
$$

Either $t_\gamma c_\gamma\leq (la)'$ for all
$\gamma\in{\mathcal E}$ and hence $z\in [0, (la)']$
or there exists $\gamma_0\in{\mathcal E}$ such that
$t_{\gamma_0} c_{\gamma_0}\not\leq (la)'$.
Hence, by almost orthogonality, either
 $j_{k_0}b_{k_0}\leq t_{\gamma_0} c_{\gamma_0}\leq z$
for some $k_0\in \{1, \dots, n\}$ or
$(n_a+1-l)a\leq t_{\gamma_0} c_{\gamma_0}\leq z$.
 In both cases we get that  $z\in D$.

 Now, assume that
$y\in{}[0, (la)']\cap \left(\bigcup_{k=1}^{n}[j_k{}b_k, 1]\cup %
[(n_a+1-l)a, 1]\right)$. Then  $(n_a+1-l)a\leq y \leq (la)'$
or $j_k{}b_k\leq y \leq (la)'$ for some $k\in \{1, \dots, n\}$.
In any case we have a contradiction.
\end{proof}

\begin{proposition}\label{clopenint}
Let $E$ be an almost
orthogonal Archimedean atomic lattice effect algebra.
Then, for any not necessarily different atoms $a, b\in E$
and any integers $l, k$;
$1\leq l\leq n_a$,  $1\leq k\leq n_b$, the interval
$[kb,(la)']$ is clopen in the interval topology.
\end{proposition}
\begin{proof} From Lemma \ref{almostort}
we have that $[0, (la)']$ is a clopen subset. Since a dual of an
almost orthogonal Archimedean atomic lattice effect
algebra is an almost orthogonal Archimedean atomic lattice
effect algebra as well, we have that
$[kb,1]$ is again clopen in the interval topology. Hence
also $[kb,(la)']$ is clopen in the interval topology.
\end{proof}

\begin{theorem}\label{haus}
Let $E$ be an almost
orthogonal Archimedean atomic lattice effect algebra.
Then the interval topology $\tau_i$ on $E$ is Hausdorff.
\end{theorem}
\begin{proof} Let $x, y\in E$ and $x\not =y$. Then (without loss
of generality) we may assume that $x\not \leq y$. Then
by \cite[Theorem 3.3]{ZR65} there
exists an atom $b$ from $E$ and an integer $k$,
$1\leq k\leq n_b$ such that $kb\leq x$ and $kb\not\leq y$.
Applying the dual of \cite[Theorem 3.3]{ZR65}
there exists an atom $a$ from $E$ and an integer $l$,
$1\leq l\leq n_a$ such that $y\leq (la)'$ and $kb\not\leq (la)'$.
Clearly, $x\in [kb, 1]$, $y\in [0, (la)']$.

Assume that there is an element $z\in E$ such that
$z\in [kb, 1]\cap [0, (la)']$. Then $kb \leq z\leq  (la)'$,
a contradiction. Hence by Proposition \ref{clopenint}, $[kb, 1]$ and $[0, (la)']$ are disjoint
open subsets separating $x$ and $y$.
\end{proof}

\begin{theorem}\label{compger}
Let $E$ be an almost
orthogonal Archimedean atomic lattice effect algebra.
Then $E$ is compactly generated and therefore (o)-continuous.
\end{theorem}
\begin{proof} It is enough to check that, for any atom
$a\in E$ and any integer $l$,
$1\leq l\leq n_a$ the element $la$ is compact in $E$ since
any element of $E$ is a join of such elements (see
Theorem \ref{Theorem3.3} resp. \cite[Theorem 3.3]{ZR65}).

Let $x=\bigvee_{\alpha\in\mathcal{E}} x_\alpha$ for some net
$(x_\alpha)_{\alpha\in\mathcal{E}}$ in $E$, $la\leq x$, i.e.,
$(la)'\geq x'\downarrow x^{'}_\alpha$.

By Lemma \ref{almostort} we have
$E=[0, (la)']\cup \left(\bigcup_{k=1}^{n}[j_k{}b_k, 1]\cup %
[(n_a+1-l)a, 1]\right)$,
$[0, (la)']\cap \left(\bigcup_{k=1}^{n}[b_k, 1]\cup %
[(n_a+1-l)a, 1]\right)=\emptyset$,
$b_1, \dots, b_n$ are atoms of $E$,
$1\leq j_k\leq n_{b_k}$, $1\leq k \leq n$.

Since ${\mathcal{E}}$
is directed upwards, there
exists a cofinal subset
$\mathcal{E}'\subseteq\mathcal{E}$ such that
$x^{'}_{\beta}\in [0, (la)']$ for all $\beta\in\mathcal{E}'$
or there exists $k_0\in \{1,2,...,n\}$ such that
$x_{\beta}\in [j_{k_0}b_{k_0}, 1]$ for all $\beta\in\mathcal{E}'$ or
$x^{'}_{\beta}\in [(n_a+1-l)a, 1]$ for all $\beta\in\mathcal{E}'$.
If $x^{'}_{\beta}\in [0, (la)']$ for all $\beta\in\mathcal{E}'$ then
clearly $la\leq x_{\beta}$ for all $\beta\in\mathcal{E}'$.
If there exists $k_0\in \{1,2,...,n\}$ such that
$x^{'}_{\beta}\in [j_{k_0}b_{k_0}, 1]$ for all $\beta\in\mathcal{E}'$ or
$x^{'}_{\beta}\in [(n_a+1-l)a, 1]$ for all $\beta\in\mathcal{E}'$
we obtain that $x'\in [j_{k_0}b_{k_0}, 1]$ or $x^{'}\in [(n_a+1-l)a, 1]$
which is a contradiction with $x^{'}\in [0, (la)']$.
\end{proof}

Let $E$ be an Archimedean atomic lattice effect algebra.
We put ${\mathcal U}=\{x\in E\mid x=\bigvee_{i=1}^{n} l_i a_i,
a_1, \dots, a_n\ \mbox{are atoms of }\ E, 1\leq l_i\leq n_{a_i},
1\leq i\leq n, n \ \mbox{natural}$\ $\mbox{number}\}$ and
${\mathcal V}=\{x\in E\mid x'\in {\mathcal U}\}$. Then by
\cite[Theorem 3.3]{ZR65}, for every $x\in L$, we have
that
$$
x=\bigvee\{u\in {\mathcal U}\mid u\leq x\}=%
\bigwedge\{v\in {\mathcal V}\mid x\leq v\}.
$$

Consider the function family
$\Phi=\{f_{u}\mid u\ \in {\mathcal U}\}\cup
\{g_{v}\mid v\in {\mathcal V}\}$, where
$f_{u}, g_v:L\to\{0,1\}$, $u\in {\mathcal U}, v\in {\mathcal V}$
are defined by putting
$f_u(x)=\left\{\begin{array}{r c l}
                  1& \hbox{iff}&u\leq x\\
                  0& \hbox{iff}&u\not\leq x
                \end{array}
         \right.       $ \\
and
$g_{v}(y)=\left\{\begin{array}{r c l}
                  1& \hbox{iff}&x\leq v\\
                  0& \hbox{iff}&x\not\leq v
                \end{array}
         \right.$ \quad for all $x, y\in L$.

Further, consider the family of pseudometrics on $L$:
$\Sigma_\Phi=\{\rho_{u}\mid u\in {\mathcal U}\}\cup
\{\pi_{v}\mid v\in {\mathcal V}\}$, where
$\rho_{u}(a,b)=|f_{u}(a)-f_{u}(b)|$ and
$\pi_{v}(a,b)=|g_{v}(a)-g_{v}(b)|$
for all $a,b\in L$.

Let us denote by $\mathcal U_\Phi$ the uniformity
on $L$ induced by the family of pseudometrics $\Sigma_\Phi$
(see e.g. \cite{csaszar}).
Further denote by $\tau_\Phi$ the topology compatible with the
uniformity $\mathcal U_\Phi$.

Then for every net
$(x_\alpha)_{\alpha\in\mathcal E}$ of elements of $L$
$$
\begin{array}{l}
x_\alpha\tophi x \hbox{ iff }
\varphi(x_\alpha)\to \varphi(x)\hbox{ for any }
\varphi\in\Phi.
\end{array}
$$

This implies, since $f_u$, $u\in {\mathcal U}$, and $g_v$, $v\in {\mathcal V}$, is a separating function family on $L$, that the topology $\tau_{\Phi}$ is Hausdorff.
Moreover, the intervals $[u, v]=[u, 1]\cap [0, v]=f_u^{-1}(\{1\})\cap g_v^{-1}(\{1\})$
are clopen sets in  $\tau_{\Phi}$.

\begin{definition}\label{deftauphi}
Let $E$ be an Archimedean atomic lattice effect algebra. Let
$\Phi$ be a separating function family on $E$ defined above.
We will denote by $\tau_{\Phi}$ the uniform topology on $E$
defined by this function family, that means for every net
$(x_\alpha)_{\alpha\in\mathcal E}$ of elements of $L$
$$
\begin{array}{l}
x_\alpha\tophi x\ \hbox{ iff }\
\varphi(x_\alpha)\to \varphi(x)\hbox{ for any }
\varphi\in\Phi.
\end{array}
$$
\end{definition}

\begin{theorem}\label{cahrh}
Let $E$ be an almost
orthogonal Archimedean atomic lattice effect algebra.
Then $\tau_i=\tau_o=\tau_{\Phi}$.
\end{theorem}
\begin{proof} Since by Theorem \ref{haus}, $\tau_i$ is Hausdorff
we obtain by \cite{erne1} that $\tau_i=\tau_o$.
Further if $O\in\tau_o$ and $x\in O$ then by Theorem \ref{Theorem3.3}
we have
$x=\bigvee\{u\in {\mathcal U}\mid u\leq x\}=%
\bigwedge\{v\in {\mathcal V}\mid x\leq v\}$, which
by \cite{ZR19} implies that there exist finite sets $F\subseteq {\mathcal U}$,
$G\subseteq {\mathcal V}$ such that $x\in [\bigvee F, \bigwedge G]\subseteq O$.
Hence $\tau_o\subseteq \tau_{\Phi}$. By Theorem \ref{compger} and
 \cite[Theorem 1]{PR4} we obtain $\tau_o=\tau_{\Phi}$.
\end{proof}

\begin{theorem}\label{semicahrh}
Let $E$ be an Archimedean atomic block-finite lattice effect algebra. Then
 $\tau_i=\tau_o$ is a Hausdorff topology.
\end{theorem}
\begin{proof} As in \cite{ZR34}, it suffices to show that for every
$x, y\in E$, $x\not\leq y$ there are finitely many intervals, none of
which contains both $x$ and $y$ and the union of which covers $E$.

By \cite{mosna}, $E$ is a union of finitely many
atomic blocks $M_i$, $i=1, 2, \dots, n$. Choose
$i\in \{ 1, 2, \dots, n\}$. If $x, y\in M_i$ then
there is an atom $a_i\in M_i$ and an integer $l_i$, $1\leq l_i \leq n_{a_i}$
such that $l_i a_i\leq x$, $l_i a_i\not\leq y$.
Let us put $k_i=n-l_i+1$. Since $M_i$ is almost orthogonal
(the only possible non-orthogonal $kb$ to $la$ for  an atom $a$, $1\leq l\leq n_a$ is that
$a=b$)
we have by Lemma \ref{almostort} that
$M_i=([0, (k_i a_i)']\cap M_i) \cup  %
([(n_{a_i}+1-k_i)a_i, 1]\cap M_i)$.
Hence $M_i\subseteq [0, (k_i a_i)'] \cup  %
[(n_{a_i}+1-k_i)a_i, 1]$. Let us check that
$[0, (k_i a_i)'] \cap  %
[(n_{a_i}+1-k_i)a_i, 1]=\emptyset$.
Assume that $(n_{a_i}+1-k_i)a_i\leq z\leq (k_i a_i)'$. Then
$(n_{a_i}+1-k_i)a_i\leq (k_i a_i)'$, a contradiction.
Put $J_i=[0, (k_i a_i)']$, $K_i=[(n_{a_i}+1-k_i)a_i, 1]$.
This yields $x\in K_i$, $y\in J_i$, $M_i\subseteq J_i\cup K_i$ and $J_i\cap K_i=\emptyset$.
Let $x\not\in M_i$. Then there exists an atom $a_i\in M_i$
that is not compatible with $x$. Let us check that
$x\not\in [0, (a_i)'] \cup  %
[n_{a_i} a_i, 1]$. Assume that
$x\in [0, (a_i)']$ or  %
$x\in [n_{a_i} a_i, 1]$. Then
$x\leq (a_i)'$ or $a_i\leq n_{a_i} a_i \leq x$, i.e., in both
cases we get that  $x\comp a_i$, a contradiction.
Let us put $J_i=[0, (a_i)']$, $K_i=[n_{a_i}a_i, 1]$. As above,
$M_i\subseteq J_i\cup K_i$, $J_i\cap K_i=\emptyset$ and moreover
$x\notin J_i\cup K_i$. The remaining case $y\not\in M_i$ can be
checked by similar considerations. We obtain
$E=\bigcup_{i=1}^{n} M_i\subseteq \bigcup_{i=1}^{n} (J_i\cup K_i)\subseteq E$
and none of the intervals $J_i, K_i$, $i=1, 2, \dots, n$ contains both
$x$ and $y$.
\end{proof}


\section{Order and interval topologies
of complete atomic block-finite lattice effect algebras}

We are going to show that on every complete atomic block-finite
lattice effect algebra $E$ the interval topology is Hausdorff. Hence
both topologies $\tau_i$ and $\tau_o$ are in this case compact Hausdorff
and they coincide. Moreover, a necessary and sufficient condition for a complete
atomic lattice algebra $E$ to be almost orthogonal is given.

For the proof of Theorems \ref{convblokse} and \ref{blockfin}
we will use the following statement, firstly
proved in the equivalent setting of D-posets in \cite{ZR41}.

\begin{theorem}{\rm{}\cite[Theorem 1.7]{ZR41}}\label{demonst} Suppose that
$(E;\oplus, 0, 1)$ is a complete lattice effect algebra. Let
$\emptyset\not =D\subseteq E$ be a sub-lattice
effect algebra.
The following conditions are equivalent:
\begin{enumerate}
\item[\mbox{\rm(i)}]  For all nets $(x_\alpha)_{\alpha\in\mathcal{E}}$ such that
$x_\alpha\in D$ for all ${\alpha\in\mathcal{E}}$
$$
x_\alpha\too x \ \mbox{in}\ E \ \mbox{if and only if} \
x\in D\ \mbox{and}\ x_\alpha\too x \ \mbox{in}\ D.
$$
\item[\mbox{\rm(ii)}]  For every $M\subseteq D$ with
$\bigvee M=x$ in $E$ it holds $x\in D$.

\item[\mbox{\rm(iii)}]  For every $Q\subseteq D$ with
$\bigwedge Q=y$ in $E$ it holds $y\in D$.

\item[\mbox{\rm(iv)}] $D$ is a complete sub-lattice of $E$.

\item[\mbox{\rm(v)}]  $D$ is a closed set in order topology
$\tau_o$ on $E$.
\end{enumerate}
Each of these conditions implies that $\tau_o^{D}=\tau_o^{E}\cap D$,
where $\tau_o^{D}$ is an order topology on $D$.
\end{theorem}

Important sub-lattice effect algebras are blocks, $S(E)$,
$B(E)=\bigcap\{M\subseteq E\mid M\ \mbox{block\/  of}\  E\}$ and
$C(E)=B(E)\cap S(E)$ (see \cite{GrFoPu,gudder1,Kop2,ZR51,ZR56}).

\begin{theorem}\label{convblokse} Let $E$ be a complete lattice effect algebra.
 Then for every $D\in \{S(E), C(E), B(E)\}$ or $D=M$, where $M$ is a block
of $E$, we have:
\begin{enumerate}
\item[\mbox{(1)}]
\begin{tabular}{l l}
$x_\alpha\stackrel{\tau^{E}_{i}}{\rightarrow} x$
$\Longleftrightarrow$
$x_\alpha\stackrel{\tau^{D}_{i}}{\rightarrow} x$,& for all nets
$(x_\alpha)_{\alpha\in\mathcal{E}}$ in $D$ and all $x\in D$.
\end{tabular}
\item[\mbox{(2)}] If $\tau_i$ is Hausdorff then
 \begin{center}
\begin{tabular}{l l}
$x_\alpha\stackrel{\tau^{E}_{i}}{\rightarrow} x$
$\Longleftrightarrow$
$x_\alpha\stackrel{\tau^{D}_{i}}{\rightarrow} x$,& for all nets
$(x_\alpha)_{\alpha\in\mathcal{E}}$ in $D$ and all $x\in E$.
\end{tabular}
\end{center}
\end{enumerate}
\end{theorem}
\begin{proof} The first part of the statement follows
by Theorem \ref{compconvti} and
the fact that if $E$ is  a complete lattice effect algebra
 then $M$, $S(E)$, $C(E)$ and $B(E)$  are complete sub-lattices of $E$
(see \cite{ZR58,ZR60}). The second part follows by \cite{erne1} since
$\tau_i$ is Hausdorff implies $\tau_i=\tau_o$ and by Theorem \ref{demonst}.
\end{proof}

\begin{theorem}\label{blockfin} \mbox{\rm(i)} The interval topology
$\tau_i$ on every Archimedean atomic MV-effect algebra $M$ is
Hausdorff and $\tau_i=\tau_o=\tau_{\Phi}$.

\noindent{}\mbox{\rm(ii)} For every complete  atomic MV-effect algebra $M$
and for any net $(x_{\alpha})$ of $M$ and any $x\in M$,
\begin{center}
$x_{\alpha}\totau x$ if and only if
$x_{\alpha}\too x$\quad (briefly $\tau_o\equiv (o)$).
\end{center}
Moreover, $\tau_o$ is a uniform compact Hausdorff topology on $M$.

\noindent{}\mbox{\rm(iii)} For every  atomic block-finite
lattice effect algebra $E$, $E$ is a complete lattice
iff $\tau_i=\tau_o$ is a compact Hausdorff topology.
\end{theorem}
\begin{proof}\mbox{\rm(i), (ii):} This follows from the fact that every
pair of elements of $M$ is compatible, hence every pair of atoms is orthogonal.
Thus for (i) we can apply Theorem \ref{cahrh} and for (ii) we can use
(i) and \cite[Theorem 2]{PR4} since $M$ is compactly generated by finite elements
and $\tau_i$ is compact.

\noindent{}\mbox{\rm(iii)} From Theorem \ref{semicahrh}  we know that
 $\tau_i=\tau_o$ is a Hausdorff topology. By Lemma \ref{pomlem} (vi) the interval topology
$\tau_i$ on $E$ is compact iff $E$ is a complete lattice.
\end{proof}

In what follows we will need Corollary \ref{fincomp} of  Lemma \ref{compact}.

\begin{lemma}\label{compact}
Let $E$ be an  Archimedean atomic lattice effect algebra. Then
\begin{enumerate}\item[\mbox{\rm{}(i)}]  If
$c, d\in E$ are compact elements with $c\leq d'$ then $c\oplus d$ is compact.
\item[\mbox{\rm{}(ii)}] If $u=\bigoplus G$, where $G$ is a $\oplus$-orthogonal
system of  atoms of $E$, and $u$ is compact then $G$
is finite.
\end{enumerate}
\end{lemma}
\begin{proof}\mbox{\rm{}(i)}  Let $c\oplus d\leq \bigvee D$.
Let ${\cal E}=\{F\subseteq D: F \ \mbox{is finite}\}$ be directed
by set inclusion and let for every $F\in {\cal E}$ be $x_F=\bigvee F$. Then
$x_F \uparrow x=\bigvee D$. Since $c\leq \bigvee D$
and $d\leq \bigvee D$ there is a finite subset $F_1\subseteq D$ such that
$c\vee d\leq \bigvee F_1$. Therefore, for $F\geq F_1$,
$x_F\ominus c \uparrow x\ominus c, d\leq x\ominus c$. Then there is
a finite subset $F_2\subseteq D$, $F_1\leq F_2$ such that
$d\leq x_{F_2}\ominus c$. Hence $c\oplus d \leq x_{F_2}$.

\noindent{}\mbox{\rm{}(ii)} Let $u\in E$,   $u=\bigoplus G=\bigvee\{\bigoplus K\mid K\subseteq G$ is finite$\}$
where $G=(a_{\kappa})_{\kappa\in H}$ is a $\oplus$-orthogonal system of  atoms. Clearly if
$K_1, K_2\subseteq G$ are finite such that $K_1\subseteq K_2$ then
$\bigoplus K_1\leq \bigoplus K_2$.

Assume that $u$ is compact. Hence
there are finite $K_1, K_2,$ $\dots, K_n\subseteq G$ such that
$u\leq \bigvee\{\bigoplus K_i \mid i=1, 2, \dots, n\}$.
Let $K_0=\bigcup\{K_i \mid i=1, 2, \dots, n\}$.
Then $K_0\subseteq G$, $K_0$ is finite and $\bigoplus K_i\leq \bigoplus K_0$,
$i=1, 2 \dots, n$, which gives that
$\bigvee\{\bigoplus K_i \mid i=1, 2, \dots, n\}\leq \bigoplus K_0$.
It follows that $u\leq \bigoplus K_0 \leq u=\bigvee\{\bigoplus K\mid K\subseteq G$
is finite$\}$. Hence  $u= \bigoplus K_0$, $K_0\subseteq G$ is finite.
Further, for every finite $K\subseteq G\setminus K_0$ we have
$\bigoplus K_0\subseteq \bigoplus (K_0\cup K)=%
\bigoplus K_0\oplus\bigoplus K\leq  u=\bigoplus K_0$ , which gives  that
$\bigoplus K=0$. Hence $K=\emptyset$ and thus $G\setminus K_0=\emptyset$
which gives that $K_0= G$.\end{proof}

\begin{corollary}\label{fincomp} Let $E$ be an o-continuous  Archimedean atomic lattice effect algebra.
Then every finite element of $E$ is compact.
\end{corollary}
\begin{proof} Clearly, by \cite[Theorem 7]{PR4} we know that $E$ is
compactly generated. Therefore, any atom of $E$ is compact.
The compactness of every finite element follows by an easy induction.
\end{proof}

\begin{theorem}\label{efcharalmost}
Let $E$ be an Archimedean  atomic lattice effect algebra. Then the following conditions
are equivalent:
\begin{enumerate}
\item[\mbox{\rm(i)}] $\tau_i=\tau_o=\tau_{\Phi}$.
\item[\mbox{\rm(ii)}] $E$ is o-continuous and $\tau_i$ is Hausdorff.
\item[\mbox{\rm(iii)}] $E$ is almost orthogonal.
\end{enumerate}
\end{theorem}
\begin{proof} \mbox{\rm(i)}$\implies$\mbox{\rm(ii)}: Since
$\tau_o=\tau_{\Phi}$ we have by \cite[Theorem 1]{PR4} that $E$ is compactly
generated and hence o-continuous. The condition $\tau_i=\tau_{\Phi}$ implies
that $\tau_i$ is Hausdorff because $\tau_{\Phi}$ is Hausdorff.

\noindent{}\mbox{\rm(ii)}$\implies$\mbox{\rm(i)}, \mbox{\rm(iii)}:
Since $\tau_i$ is Hausdorff we obtain  $\tau_i=\tau_o$ by \cite{erne1}.
Moreover, from \cite[Theorem 7]{PR4} and Corollary \ref{fincomp}
the (o)-continuity of $E$ implies that
$E$ is compactly generated by the elements from  ${\mathcal U}$. This gives
 $\tau_o=\tau_{\Phi}$ from \cite[Theorem 1]{PR4}.

 Let $a\in E$ be an atom, $1\leq  l \leq n_a$. Then the interval
 $[0, (la)']$ is a clopen set in the order topology $\tau_o=\tau_{\Phi}=\tau_i$.
 Hence there is a finite set of intervals in $E$ such that
 $0\in E\setminus \bigcup_{i=1}^{n} [u_i, v_i]\subseteq [0, (la)']$.
 Thus $E\subseteq [0, (la)']\cup \bigcup_{i=1}^{n} [u_i, v_i]\subseteq
[0, (la)']\cup \bigcup_{i=1}^{n} [k_i{}b_i, 1]$, where $b_i\in E$ are
atoms such that $k_i{}b_i\leq u_i$, $1\leq k_i\leq n_{b_i}$, $i=1, \dots, n$. This yields that
$E$ is almost orthogonal.

\noindent{}\mbox{\rm(iii)}$\implies$\mbox{\rm(ii)}: From Theorems \ref{haus}  and \ref{compger}
we have that $\tau_i$ is Hausdorff and $E$ is compactly generated, hence (o)-continuous.
 \end{proof}

\begin{corollary}\label{charalmost}
Let $E$ be a complete  atomic lattice effect algebra. Then the following conditions
are equivalent:
\begin{enumerate}
\item[\mbox{\rm(i)}] $E$ is almost orthogonal.
\item[\mbox{\rm(ii)}] $\tau_i=\tau_o=\tau_{\Phi}\equiv  (o)$.
\item[\mbox{\rm(iii)}] $E$ is (o)-continuous and $\tau_i$ is Hausdorff.
\end{enumerate}
\end{corollary}
\begin{proof}  It follows from
Theorems \ref{efcharalmost} and the fact that by (o)-continuity of $E$ \cite[Theorem 8]{ZR73} we have
$\tau_o\equiv  (o)$.
 \end{proof}

 The next example shows that a complete block-finite
 atomic lattice effect algebra need not be (o)-continuous and
 almost orthogonal in spite of that $\tau_i=\tau_o$ is a compact Hausdorff
 topology.

 \begin{example}{\rm Let $E$ be a horizontal sum of finitely many infinite complete
 atomic Boolean algebras $(B_i,\bigoplus_i, 0_i, 1_i)$,
 $i=1,2,...,n$. Then  $E$ is an atomic complete lattice effect algebra,
 $E$ is not almost orthogonal, $E$ is not
{compactly generated by finite elements} (hence $\tau_o\not =\tau_{\Phi}$),
 $E$ is  block-finite, $\tau_i=\tau_o$ is Hausdorff by Theorem \ref{semicahrh},
 and the interval topology $\tau_i$ on $E$  is compact. 
 }
 \end{example}

\section{Applications}

\begin{theorem}\label{scomport}
Let $E$ be a block-finite complete  atomic lattice effect algebra. Then the following conditions
are equivalent:
\begin{enumerate}
\item[\mbox{\rm(i)}] $E$ is almost orthogonal.
\item[\mbox{\rm(ii)}] $E$ is compactly generated.
\item[\mbox{\rm(iii)}] $E$ is (o)-continuous.
\item[\mbox{\rm(iv)}] $\tau_i=\tau_o=\tau_{\Phi}\equiv  (o)$.
\end{enumerate}
\end{theorem}
\begin{proof} By Theorem \ref{semicahrh}, $\tau_i=\tau_o$ is a Hausdorff topology.
 This by \cite[Theorem 7]{PR4} gives that
\mbox{\rm(ii)} $\Longleftrightarrow $ \mbox{\rm(iii)} and by Corollary
 \ref{charalmost} we obtain that \mbox{\rm(i)} $\Longleftrightarrow $ \mbox{\rm(iii)}
 $\Longleftrightarrow $ \mbox{\rm(iv)}.
 \end{proof}

In Theorem \ref{scomport}, the assumption that $E$
is atomic can not be omitted. For instance,
every non-atomic complete Boolean algebra is
(o)-continuous but it is not compactly
generated, because in such a case
$E$ must be atomic by \cite[Theorem 6]{PR4}.

\begin{remark}{\rm\label{cannot}
If a $\oplus$-operation on a lattice effect algebra $E$ is
continuous with respect to its interval topology $\tau_i$
meaning that $x_{\alpha}\leq y_{\alpha}^{'}$,
$x_\alpha\stackrel{\tau_{i}}{\rightarrow} x$,
$y_\alpha\stackrel{\tau_{i}}{\rightarrow} y$,
$\alpha\in\mathcal{E}$ implies
$x_\alpha\oplus y_\alpha\stackrel{\tau_{i}}{\rightarrow} x\oplus y$,
then $\tau_i$ is Hausdorff (see \cite{wulei}).
Hence $\oplus$-operation on complete (o)-continuous  atomic
lattice effect algebras which are not almost orthogonal
cannot be $\tau_i$-continuous, by \cite{wulei} and Corollary \ref{charalmost}.   }
\end{remark}

\begin{theorem}\label{ocontcomp}
Let $E$ be a block-finite  complete atomic lattice effect algebra.
Let  $(x_\alpha)_{\alpha\in\mathcal{E}}$  and
$(y_\alpha)_{\alpha\in\mathcal{E}}$ be nets of elements of $E$
such that $x_{\alpha}\leq y_{\alpha}^{'}$ for all ${\alpha\in\mathcal{E}}$.

If $x_\alpha\stackrel{\tau_{i}}{\rightarrow} x$,
$y_\alpha\stackrel{\tau_{i}}{\rightarrow} y$,
$\alpha\in\mathcal{E}$ then $x\leq y^{'}$ and
$x_\alpha\oplus y_\alpha\stackrel{\tau_{i}}{\rightarrow} x\oplus y$,
$\alpha\in\mathcal{E}$.
Moreover, $\tau_i=\tau_o$.
\end{theorem}
\begin{proof} Since, by Theorem \ref{semicahrh},
$\tau_i$ is Hausdorff, we obtain that $\tau_i=\tau_o$ by \cite{erne1}.
Let $\{ M_1, \dots, M_n\}$ be the set of all blocks of $E$.
Further, for every  ${\alpha\in\mathcal{E}}$, elements of the set
$\{ x_\alpha, y_\alpha, x_\alpha\oplus y_\alpha\}$ are pairwise
compatible. It follows that for every  ${\alpha\in\mathcal{E}}$
there exists a block $M_{k_{\alpha}}$ of $E$, ${k_{\alpha}}\in \{ 1, \dots, n\}$
such that $\{ x_\alpha, y_\alpha, x_\alpha\oplus y_\alpha\}%
\subseteq M_{k_{\alpha}}$. Let ${\mathcal{E}}'$  be
any cofinal subset of ${\mathcal{E}}$. Since ${\mathcal{E}}'$
is directed upwards, there exists a block $M_{k_{0}}$ of $E$
and a cofinal subset ${\mathcal{E}}''$ of  ${\mathcal{E}}'$  such that
$\{ x_\beta, y_\alpha, x_\beta\oplus y_\beta\}%
\subseteq M_{k_{0}}$ for all $\beta\in {\mathcal{E}}''$.
Otherwise we obtain a contradiction with the finiteness of the
set $\{ M_1, \dots, M_n\}$. Further, by Theorem \ref{compconvti},
we obtain that $\tau^{M_{k_{0}}}_{i}=\tau_i\cap M_{k_{0}}$,
as $M_{k_{0}}$ is a complete sublattice of $E$ (see Theorem \ref{convblokse}).
It follows that the interval topology $\tau^{M_{k_{0}}}_{i}$ on the
complete MV-effect algebra $M_{k_{0}}$ is Hausdorff.
The last by \cite[Theorem 3.6]{wulei} gives that
$x_\beta\oplus y_\beta\stackrel{\tau^{M_{k_{0}}}_{i}}{\rightarrow} x\oplus y$,
$\beta\in\mathcal{E}''$ and hence
$x_\beta\oplus y_\beta\stackrel{\tau_{i}}{\rightarrow} x\oplus y$,
$\beta\in\mathcal{E}''$, as $\tau^{M_{k_{0}}}_{i}=\tau_i\cap M_{k_{0}}$.
It follows that $x_\alpha\oplus y_\alpha\stackrel{\tau_{i}}{\rightarrow} x\oplus y$,
$\alpha\in\mathcal{E}$ by Lemma \ref{charcofin}.
\end{proof}

In \cite[Theorem 4.5]{ZR54} it was proved that
a block-finite lattice effect algebra $(E;\oplus,0,1)$ has a
MacNeille completion which is a complete effect algebra
$(MC(E);\oplus,0,1)$ containing $E$ as a
(join-dense and meet-dense) sub-lattice effect algebra
 iff $E$ is Archimedean. In what follows we put $\widehat E=MC(E)$.

\begin{corollary}\label{ocontarch}
 Let $E$ be a block-finite  Archimedean atomic lattice effect algebra.
 Then for any nets  $(x_\alpha)_{\alpha\in\mathcal{E}}$  and
$(y_\alpha)_{\alpha\in\mathcal{E}}$ of elements of $E$
with $x_{\alpha}\leq y_{\alpha}^{'}$, ${\alpha\in\mathcal{E}}$:
$x_\alpha\stackrel{\tau_{i}}{\rightarrow} x$,
$y_\alpha\stackrel{\tau_{i}}{\rightarrow} y$,
$\alpha\in\mathcal{E}$ implies
$x_\alpha\oplus y_\alpha\stackrel{\tau_{i}}{\rightarrow} x\oplus y$,
$\alpha\in\mathcal{E}$.
\end{corollary}
\begin{proof} By \cite[Lemma 1.1]{ZR47}, for interval topologies
$\widehat{\tau}_{i}$ on $\widehat E$ and
${\tau_{i}}$ on $E$, we have $\widehat{\tau_{i}}\cap E={\tau_{i}}$.
Thus for $x_\alpha, y_\alpha, x, y\in E$ we obtain
$x_\alpha\oplus y_\alpha\stackrel{\widehat{\tau}_{i}}{\rightarrow} x\oplus y$,
$\alpha\in\mathcal{E}$ which gives
$x_\alpha\oplus y_\alpha\stackrel{{\tau}_{i}}{\rightarrow} x\oplus y$,
$\alpha\in\mathcal{E}$ by the fact that $\widehat{\tau_{i}}\cap E={\tau_{i}}$.
\end{proof}

\begin{definition}{Let $E$ be a  lattice.
Then
\begin{enumerate}
\item[\mbox{\rm{}(i)}] An element $u$ of $E$ is called {\em strongly compact}
(briefly {\em s-compact}) iff, for any $D\subseteq E$: $u\leq c\in E$
for all $c\geq D$  implies
$u\leq \bigvee F$ for some finite $F\subseteq D$.

\item[\mbox{\rm{}(ii)}]  $E$ is called {\em s-compactly generated} iff every
element of $E$ is a join of s-compact elements.
\end{enumerate}}%
\end{definition}

\begin{theorem}\label{equivscomport}
Let $E$ be a block-finite Archimedean atomic lattice effect algebra.
Then the following conditions
are equivalent:
\begin{enumerate}
\item[\mbox{\rm(i)}] $E$ is almost orthogonal.
\item[\mbox{\rm(ii)}] $\widehat{E}=MC(E)$ is almost orthogonal.
\item[\mbox{\rm(iii)}] $\widehat{E}=MC(E)$ is compactly generated.
\item[\mbox{\rm(iv)}] ${E}$ is s-compactly generated.
\end{enumerate}
\end{theorem}
\begin{proof} By J.~Schmidt \cite{schmidt} a {MacNeille completion} $\widehat{E}$
of  $E$ is (up to isomorphism) a complete lattice such that  for every
element $x\in \widehat{E}$ there exist $P,Q\subseteq E$ such that
$x=\bigvee_{\widehat{E}} P=\bigwedge_{\widehat{E}} Q$ (taken in $_{\widehat{E}}$).
Here we identify $E$ with
$\varphi(E)$, where $\varphi:E\to {\widehat{E}}$ is the embedding
(meaning that $E$ and $\varphi(E)$ are isomorphic lattice effect algebras).
It follows that $E$ and $ {\widehat{E}}$ have the same set of all atoms and coatoms
and hence also the same set of all finite and cofinite elements,
which implies that
\mbox{\rm(i)} $\Longleftrightarrow $ \mbox{\rm(ii)}.

Moreover, for any $A\subseteq E$ and $u\in E$, we have
($d\in E$, $A\leq d$ implies  $u\leq d$) iff $u\leq \bigvee_{\widehat{E}} A$.
Then $u$ is s-compact in $E$ iff $u$ is compact in $\widehat{E}$,
which gives \mbox{\rm(iii)} $\Longleftrightarrow $ \mbox{\rm(iv)}.

Finally \mbox{\rm(ii)} $\Longleftrightarrow $ \mbox{\rm(iii)} by
Theorem \ref{scomport}.
\end{proof}

\begin{definition}{ \label{Dstate}
Let $E$ be an effect algebra. A map $\omega:E\to[0,1]$ is called a
{\em state} on $E$ if $\omega(0)=0$, $\omega(1)=1$ and $\omega(x\oplus
y)=\omega(x)+\omega(y)$ whenever $x\oplus y$ exists in $E$.}
\end{definition}

\begin{theorem}\label{statesmt} {\rm{}(State smearing theorem for
almost orthogonal block-finite Archimedean
atomic lattice effect algebras)}
Let $(E;\oplus,0,1)$ be a block-finite Archimedean atomic lattice
effect algebra. If $E$ is almost orthogonal then:
\begin{enumerate}
\item[\hbox{\rm{(i)}}] $E_1=\{x\in E\mid x\ \mbox{or}\ x'\ \mbox{is finite}\}$
is a sub-lattice effect algebra of $E$.
\item[\hbox{\rm{(ii)}}] If there exists an $(o)$-continuous state $\omega$ on
$E_1$ (or on $S(E_1)=S(E)\cap E_1$, or on $S(E)$) then
there exists an $(o)$-continuous state $\widetilde{\omega}$
on $E$ extending $\omega$ and an $(o)$-continuous state $\widehat{\omega}$
on $\widehat{E}=MC(E)=MC(E_1)$ extending $\widetilde{\omega}$.
\end{enumerate}
\end{theorem}
\begin{proof} \mbox{\rm{(i)}} By Theorem \ref{equivscomport},
${E}$ is s-compactly generated and thus by \cite[Theorem 2.7]{PR3}
$E_1$ is a sub-lattice effect algebra of $E$.

\noindent{}\mbox{\rm{(ii)}} Since ${E}$ is s-compactly generated, we obtain the
existence of $(o)$-continuous extensions $\widetilde{\omega}$
on $E$ and  $\widehat{\omega}$ on $\widehat{E}$ by \cite[Theorem 4.2]{PR3}.
\end{proof}

\end{document}